\pdfoutput=1
\documentclass{article}
\usepackage{arxiv}
\usepackage[utf8]{inputenc}
\usepackage[T1]{fontenc}
\usepackage{amsmath,amssymb,amsthm}
\usepackage{mathtools}
\usepackage{enumitem}
\usepackage{graphicx}
\usepackage{booktabs}
\usepackage{url}
\usepackage[round,semicolon,authoryear]{natbib}
\usepackage{hyperref}
\hypersetup{colorlinks=true, urlcolor=blue, linkcolor=blue, citecolor=blue}
\usepackage{doi}
\makeatletter
\def\@doi#1{%
  \let\#\relax \let\_\relax \let\textless\relax \let\textgreater\relax
  \edef\x{\toks0={{\doiurl#1}}}\x
  \edef\#{\@percentchar23}\edef\_{_}%
  \edef\textless{\@percentchar3C}\edef\textgreater{\@percentchar3E}%
  \edef\x{\toks2={\noexpand\href{\doiurl#1}}}\x
  \edef\x{\endgroup\the\toks2 \the\toks0}\x}
\makeatother
\renewcommand{\textsc}[1]{\leavevmode{\normalfont\scshape #1}}
\providecommand{\Description}[1]{}

\theoremstyle{definition}
\newtheorem{definition}{Definition}
\newtheorem{problem}{Problem}

\newtheorem{remark}{Remark}

\theoremstyle{plain}
\newtheorem{theorem}{Theorem}
\newtheorem{proposition}{Proposition}
\newtheorem{lemma}{Lemma}
\newtheorem{corollary}{Corollary}

\newtheorem{observation}{Observation}

\newcommand{\credits}{w}
\newcommand{\taken}{\mathrm{taken}}
\newcommand{\pre}{\mathrm{pre}}
\newcommand{\ttd}{\mathrm{ttd}}
\newcommand{\load}{\mathrm{load}}
\newcommand{\cl}{\mathrm{cl}}

\newcommand{\N}{\mathbb{N}}
\newcommand{\Zpos}{\mathbb{Z}_{>0}}
\newcommand{\Znn}{\mathbb{Z}_{\ge 0}}
\newcommand{\eqdef}{\vcentcolon=}

\title{The Plan Language of a Curriculum: A Formal Model and the Complexity of Degree Planning}
\author{
 Sherzod Turaev\thanks{Corresponding author.} \\
  Department of Computer Science and Software Engineering \\
  College of Computing and Artificial Intelligence \\
  United Arab Emirates University, Al Ain, United Arab Emirates \\
  \texttt{sherzod@uaeu.ac.ae} \\
  \And
  Mary John \\
  Academic Support Department \\
  Abu Dhabi Polytechnic, Abu Dhabi, United Arab Emirates \\
  \texttt{mary.john@actvet.gov.ae} \\
  \And
  Mamoun Awad \\
  Department of Computer Science and Software Engineering \\
  College of Computing and Artificial Intelligence \\
  United Arab Emirates University, Al Ain, United Arab Emirates \\
  \texttt{mamoun.awad@uaeu.ac.ae} \\
}
\date{}

\begin{document}
\maketitle

\begin{abstract}
We model an academic curriculum as a generator of a language of feasible study plans: prerequisites are monotone Boolean formulas in conjunctive normal form, degree requirements are credit-threshold covering constraints, and a study plan is a sequence of terms bounded by a per-term credit capacity. Within this model we settle the complexity of the two natural planning objectives, the number of terms to a degree and the total credit load, and we isolate the structural commitment responsible for each source of hardness. Time to degree is polynomial whenever the per-term capacity is unbounded, for arbitrary disjunctive prerequisites and arbitrary electives, so disjunction never contributes to its hardness, yet it becomes strongly NP-hard as soon as capacity binds, even without any prerequisite. Load is complementary: disjunction and overlapping electives are each strongly NP-hard in isolation and their complexity does not depend on capacity, while load is polynomial on the conjunctive, mandatory fragment. The two objectives therefore have disjoint sources of hardness. We show that the delay-factor component of the standard curricular-complexity metric is a polynomially computable upper bound on time to degree, exact on the conjunctive fragment and loose elsewhere by a quantity we name the disjunctive slack, and we prove that program subsumption is coNP-complete and consensus prerequisite recovery is NP-complete. Instantiating the model on a corpus of twenty-two universities, we find that 88 percent of prerequisite-bearing courses are purely conjunctive and that capacity, not prerequisite logic, is the operative constraint on time to degree. The curriculum corpus is openly available (\url{https://doi.org/10.5281/zenodo.22334674}), and the analysis and figure-generation code accompany the paper.
\end{abstract}


\section{Introduction}\label{sec:intro}

The structure of an academic program, the web of prerequisites, requirements, and permitted sequences through which a student must pass, is among the most consequential artifacts a university produces, and the field of \emph{curricular analytics} has made that structure an object of measurement. In the line of work established by Slim, Heileman, and their collaborators, a curriculum is modeled as a directed graph whose vertices are courses and whose arcs are prerequisite relations, and from that graph a family of structural scores, chiefly the \emph{delay factor}, the \emph{blocking factor}, and their aggregate \emph{curricular complexity}, is derived and used to compare programs, to anticipate the courses at which students accumulate, and to guide reform \citep{slim2014,heileman2018}. These scores are now computed as a matter of routine, reported by institutions, and treated as summaries of how demanding a program is to complete.

The graph on which these scores rest is, however, a simplification of the artifact it represents, and the simplification suppresses exactly the features that a printed catalog takes care to state. A prerequisite in a catalog is seldom a single course; it is a monotone Boolean condition, a conjunction of requirements in which individual clauses offer alternatives, so that a course may demand one of several equivalent predecessors together with a second course outright. A degree requirement is seldom a fixed list of mandatory courses; it is a \emph{covering constraint} that asks for a stated number of credits from a designated pool, leaving the student a choice of electives. And a path through a program is not a topological order of the whole graph but a sequence of \emph{terms}, each bounded by a credit ceiling and by which courses a term actually offers. A parallel body of work has engaged these features operationally rather than analytically: degree-plan planners such as the KnowPlan system compile catalog data into a structured plan representation and search for individual degree pathways \citep{cao2026}, and learning-analytics models such as the CAPIRE curriculum graph engineer structural features from prerequisite networks to model student progression \citep{paz2025}. The classical theory of scheduling under precedence constraints, and the combinatorial theory of knowledge spaces, each supply pieces of the relevant mathematics \citep{ullman1975,gl1995,df1985}.

Though each of these lines is mature, they have not been joined, and a prior question has gone unasked. No account treats a curriculum as a formal object in its own right, defines the language of study plans that complete it, and asks which questions about that language are decidable, which are tractable, and where real curricula fall relative to the boundary between them. The omission is not merely foundational tidiness. Our companion empirical study \citep{turaev2026} shows that at the institutions whose catalogs publish the most alternatives, between a third and a half of the prerequisite-bearing courses carry a disjunction that the all-conjunctive reading, on which every complexity score depends, silently treats as an obligation, while at the institutions that publish none the two readings coincide. The metrics, in other words, rest on an assumption that the data contradict, yet there is no theory that explains the discrepancy, no result that says which structural commitment is responsible for the difficulty of planning a degree, and no measurement of whether the hard cases the mathematics admits are cases that actually occur.

To overcome these drawbacks, in the current research we treat a curriculum as a generator of a formal language, namely the language of the feasible study plans that complete it, and we determine the computational character of the questions this formulation makes precise. We define the model so that its only commitments beyond a precedence graph are the three the catalog states, the Boolean form of prerequisites, the covering form of requirements, and the per-term capacity, and we then attribute the difficulty of degree planning to those commitments individually. The main contributions of the paper are the following.

\begin{enumerate}[label=(\arabic*),nosep,leftmargin=2.2em]
  \item We introduce the \emph{plan language} of a curriculum and establish its basic theory: recognition of a plan is linear, and on the conjunctive fragment the shortest plan is the longest prerequisite chain, computable in linear time.
  \item We prove that the time to degree is computable in polynomial time whenever the per-term capacity is unbounded, for arbitrary disjunctive prerequisites and arbitrary electives, and that it becomes strongly NP-hard the moment capacity binds, already in the absence of any prerequisite; disjunction, contrary to the natural expectation, never contributes to the hardness of this objective.
  \item We prove that the load objective is complementary, being polynomial on the conjunctive \emph{and mandatory} fragment yet strongly NP-hard through either of two independent commitments, disjunction or overlapping electives, each in isolation and each independent of capacity, and we draw from this the paper's organizing result, that the two planning objectives have \emph{disjoint} sources of hardness, capacity for time and, for load, disjunction or electives.
  \item We exhibit the delay-factor component of the standard curricular-complexity metric as a polynomially computable upper bound on the true time to degree, exact on the conjunctive fragment and loose elsewhere by a quantity we call the \emph{disjunctive slack}, thereby giving the empirical correction of our companion study \citep{turaev2026} a formal explanation.
  \item We settle the model's two comparison problems, showing that deciding whether one program subsumes another is coNP-complete and that recovering a consensus prerequisite relation across several curricula is NP-complete, each under strong structural restrictions, and we instantiate the entire model on a corpus of twenty-two universities, locating real curricula on the tractable side of every boundary but capacity.
\end{enumerate}

The remainder of the paper is organized as follows. Section~\ref{sec:related} reviews the five bodies of work the formulation draws together and states precisely how the present contribution differs from each. Section~\ref{sec:prelim} fixes the notation and the standard facts. Section~\ref{sec:model} defines the model, the plan language, and the two objectives, and Section~\ref{sec:recognition} settles recognition. Section~\ref{sec:timefree-sec} proves the tractability of time to degree without a capacity bound, and Section~\ref{sec:capacity} its hardness once capacity binds. Section~\ref{sec:load} treats the load objective and establishes the hardness duality. Section~\ref{sec:metric-sec} connects the model to the curricular-complexity metric, and Section~\ref{sec:empirical} reports the empirical boundary across the corpus. Section~\ref{sec:compare} settles the two comparison problems, and Section~\ref{sec:summary} concludes.

\section{Related work}\label{sec:related}

\paragraph{Curricular analytics and complexity metrics.} The measurement of curricular structure through graph-theoretic scores originates with the cruciality, delay, and blocking factors of \citet{slim2014} and the aggregate complexity framework of \citet{heileman2018}, which has since been applied to program comparison and reform across many institutions. That framework is, by construction, a functional of the prerequisite graph under the all-conjunctive reading, and it does not model disjunctive prerequisites, credit-threshold electives, or term capacity. The present work does not propose a new score. It supplies the formal model in which the prerequisite graph underlying such scores is a single reading, the conjunctive collapse, and this allows us to state in Section~\ref{sec:metric-sec} precisely what the delay-factor component of the metric measures and under which structural condition it overstates the time to degree.

\paragraph{Degree-plan optimization.} A separate line treats the construction of an individual student's plan as an optimization problem, most recently the KnowPlan system, which compiles catalog data into a structured representation and searches for personalized degree pathways \citep{cao2026}. This work is solution-oriented: it computes good plans for given instances, and it neither classifies the complexity of the underlying questions nor asks where real curricula sit relative to tractability. Our contribution is complementary and logically prior to it, in that we characterize which planning questions are tractable and identify the single structural commitment, capacity, that renders the time objective intractable and that any such optimizer must therefore confront.

\paragraph{Scheduling under precedence.} The hardness results we prove draw on the classical theory of precedence-constrained scheduling. Minimizing the makespan of unit jobs under a partial-order precedence on parallel machines is NP-hard when the number of machines is part of the input \citep{ullman1975}, and scheduling under AND/OR precedence has been studied in its own right \citep{gl1995,mss2004}. We use these results as the source of the capacity hardness in Section~\ref{sec:capacity}, and we contribute the observation, which the scheduling literature does not frame in curricular terms, that the AND/OR structure of prerequisites is by itself tractable for the time objective, so that, for that objective, the hardness is carried entirely by the capacity constraint; elective covering enters only the load objective.

\paragraph{Knowledge spaces.} The combinatorial theory of knowledge spaces \citep{df1985,fd2011} studies the states of knowledge reachable under a prerequisite relation and connects them to closure systems and lattices, and it is the closest existing formalization of prerequisite structure. It is developed, however, at the level of assessment items rather than of credit-bearing courses scheduled into capacity-bounded terms, and it does not address the optimization and comparison questions that a degree, as opposed to a knowledge state, raises. Our plan language may be seen as a scheduling-aware and requirement-aware extension of that combinatorial picture.

\paragraph{Prediction from curriculum structure.} Finally, a learning-analytics line uses structural features of the prerequisite network as predictors of student progression, as in the CAPIRE curriculum graph \citep{paz2025}. Such models are empirical and predictive rather than analytical, and they take the structural features as given. The present work instead asks what those features mean computationally, and it supplies the exact quantities, time to degree, load, and the disjunctive slack, that a predictive model of progression would ideally rest upon.

\section{Preliminaries: notation and notions}\label{sec:prelim}

This section fixes the notation and collects the standard definitions and results on which the later sections rely. The two observations follow immediately from the definitions; the remaining results are classical and are stated without proof, with references given.

\subsection{Sets, weights, and monotone Boolean formulas}

We write $\N = \Znn = \{0,1,2,\dots\}$ and $\Zpos = \{1,2,\dots\}$. For a finite set $V$ the power set is $2^{V}$ and the cardinality is $|V|$. Given a weight map $\credits : V \to \Zpos$ and a subset $X \subseteq V$ we abbreviate $\credits(X) = \sum_{v \in X} \credits(v)$, with $\credits(\emptyset) = 0$.

A \emph{Boolean formula} over the variable set $V$ is \emph{monotone} if it is built from the variables using only conjunction ($\wedge$) and disjunction ($\vee$), with no negation \citep{cramahammer}. It is in \emph{conjunctive normal form} (CNF) if it is a conjunction of \emph{clauses}, each clause a disjunction of variables. A subset $X \subseteq V$ induces the truth assignment that sets $v$ to true exactly when $v \in X$; we write $X \models \varphi$ when this assignment satisfies $\varphi$, and adopt the convention that the empty conjunction is $\top$ (satisfied by every $X$) and the empty disjunction is $\bot$ (satisfied by no $X$).

\begin{observation}[Clause-wise satisfaction]\label{fact:clause}
For a CNF formula $\varphi = \bigwedge_{j=1}^{k}\big(\bigvee_{u \in A_j} u\big)$ and any $X \subseteq V$, we have $X \models \varphi$ if and only if $A_j \cap X \neq \emptyset$ for every $j \in \{1,\dots,k\}$.
\end{observation}

\begin{observation}[Monotonicity of satisfaction]\label{fact:monotone}
If $\varphi$ is a monotone Boolean formula, $X \models \varphi$, and $X \subseteq Y \subseteq V$, then $Y \models \varphi$.
\end{observation}

Both observations follow directly from the definitions, the second by induction on the structure of $\varphi$ \citep{cramahammer}.

\subsection{Digraphs, orders, and longest paths}

A \emph{directed graph} (digraph) is a pair $(V,E)$ with $E \subseteq V \times V$. For $v \in V$ its \emph{in-neighbours} are $\{u : (u,v) \in E\}$; a \emph{source} has none. A \emph{directed path} is a sequence $v_1, v_2, \dots, v_t$ of distinct vertices with $(v_s, v_{s+1}) \in E$ for all $s$; its \emph{length} is the number of vertices $t$ on it. The digraph is a \emph{directed acyclic graph} (DAG) if it contains no directed cycle. A \emph{topological order} is a linear order of $V$ in which every arc goes forward.

\begin{proposition}[Topological order and longest paths \citep{clrs}]\label{fact:topo}
A digraph admits a topological order if and only if it is a DAG. Given a DAG $(V,E)$ and a weight on vertices, a topological order, and the length of a longest directed path ending at each vertex, can all be computed in time $O(|V| + |E|)$ by a single dynamic-programming pass in topological order.
\end{proposition}

\subsection{Computational complexity}

We assume the standard notions of decision problem and of the classes $\mathrm{P}$, $\mathrm{NP}$, and $\mathrm{coNP}$. A \emph{polynomial-time many-one (Karp) reduction} from a problem $A$ to a problem $B$, written $A \le_p B$, is a polynomial-time map sending yes-instances of $A$ to yes-instances of $B$ and no-instances to no-instances. A problem $B$ is \emph{NP-hard} if $A \le_p B$ for every $A \in \mathrm{NP}$, and \emph{NP-complete} if additionally $B \in \mathrm{NP}$; \emph{coNP-hardness} is defined dually.

\begin{proposition}[Transfer of hardness \citep{gj1979}]\label{fact:transfer}
If $A \le_p B$ and $A$ is NP-hard, then $B$ is NP-hard.
\end{proposition}

When an instance carries integers, its \emph{size} is measured with those integers written in binary. A problem is \emph{strongly NP-hard} if it remains NP-hard when every integer in the instance is bounded by a polynomial in the instance size, equivalently when the integers are written in unary \citep{gj1979}.

\begin{proposition}[Strong NP-hardness precludes pseudo-polynomial algorithms \citep{gj1979}]\label{fact:strong}
If a problem is strongly NP-hard, then it admits no algorithm running in time polynomial in the instance size and the largest integer value unless $\mathrm{P} = \mathrm{NP}$.
\end{proposition}

\subsection{Reference problems}

The reductions in Sections~\ref{sec:capacity} and~\ref{sec:load} use the following classical problems.

\begin{definition}[\textsc{Bin Packing}]\label{def:binpacking}
Given item sizes $s_1,\dots,s_N \in \Zpos$, a bin capacity $\kappa \in \Zpos$, and a bound $B \in \Zpos$, decide whether the items can be partitioned into at most $B$ subsets each of total size at most $\kappa$.
\end{definition}

\begin{proposition}[Hardness of \textsc{Bin Packing} \citep{gj1979}]\label{fact:binpacking}
\textsc{Bin Packing} is strongly NP-hard.
\end{proposition}

\begin{definition}[\textsc{Set Cover}]\label{def:setcover}
Given a finite universe $U$, a family $F_1,\dots,F_s \subseteq U$ with $\bigcup_i F_i = U$, and an integer $t$, decide whether there is an index set $\mathcal{K}$ with $|\mathcal{K}| \le t$ and $\bigcup_{i \in \mathcal{K}} F_i = U$.
\end{definition}

\begin{proposition}[Hardness of \textsc{Set Cover} \citep{karp1972}]\label{fact:setcover}
\textsc{Set Cover} is NP-hard.
\end{proposition}

\begin{proposition}[Precedence-constrained scheduling \citep{ullman1975,coffman1972,lenstra1978}]\label{fact:scheduling}
The problem $P \mid \mathrm{prec},\, p_j = 1 \mid C_{\max}$, of scheduling unit-length jobs subject to a partial-order precedence constraint on $m$ identical parallel machines so as to minimize the makespan, is NP-hard when the number of machines $m$ is part of the input, and deciding whether the makespan is at most three is itself NP-complete \citep{lenstra1978}; for two machines it is solvable in polynomial time \citep{coffman1972}, while its complexity for each fixed $m \ge 3$ is a long-standing open problem.
\end{proposition}

\section{The model}\label{sec:model}

Throughout, $V$ is a finite set whose elements are called \emph{courses}, and $n \eqdef |V|$.

\begin{definition}[Prerequisite formula]\label{def:prereq}
A \emph{prerequisite formula} over $V$ is a monotone Boolean formula in conjunctive normal form,
\[
  P \;=\; \bigwedge_{j=1}^{k} \Big( \bigvee_{u \in A_j} u \Big), \qquad A_j \subseteq V,\ A_j \neq \emptyset,
\]
with the convention that the empty conjunction ($k=0$) is the constant $\top$. The sets $A_j$ are its \emph{clauses}. Satisfaction $X \models P$ for $X \subseteq V$ is as defined in Section~\ref{sec:prelim}, and by Observation~\ref{fact:clause} it amounts to $A_j \cap X \neq \emptyset$ for every $j$; in particular $X \models \top$ for all $X$.
\end{definition}

By Observation~\ref{fact:monotone} the satisfaction of a prerequisite formula is monotone in the taken set: enlarging the set of completed courses never invalidates a prerequisite already met.

\begin{definition}[Curriculum]\label{def:curriculum}
A \emph{curriculum} is a tuple $C = (V, \credits, P, R, \kappa)$ where
\begin{enumerate}[label=(\roman*),nosep]
  \item $\credits : V \to \Zpos$ assigns a \emph{credit weight} to each course; for $X \subseteq V$ we write $\credits(X) \eqdef \sum_{v \in X} \credits(v)$;
  \item $P$ assigns to each course $v$ a prerequisite formula $P(v)$ over $V$, with clauses $A_{v,1},\dots,A_{v,k_v}$;
  \item $R = \{ (S_1,r_1),\dots,(S_p,r_p) \}$ is a finite set of \emph{requirement groups}, each $(S_g, r_g)$ with $S_g \subseteq V$ and $0 \le r_g \le \credits(S_g)$;
  \item $\kappa \in \Zpos \cup \{\infty\}$ is the \emph{per-term capacity}.
\end{enumerate}
The \emph{support digraph} of $C$ is $D(C) = (V, E)$ with $(u,v) \in E$ iff $u \in \bigcup_{j} A_{v,j}$, that is, iff $u$ occurs in the prerequisite of $v$. We require $D(C)$ to be acyclic, and call $C$ \emph{well-founded}. The \emph{disjunctive width} is $\omega_\vee(C) \eqdef \max_{v,j} |A_{v,j}|$ and the \emph{conjunctive width} is $\omega_\wedge(C) \eqdef \max_v k_v$.
\end{definition}

Mandatory courses and free electives are the two extremes of a requirement group: a group $(\{v\}, \credits(v))$ forces $v$ to be taken, while a group $(S,r)$ with $|S|$ large and $r$ small requests any $r$ credits from a pool.

\begin{definition}[Study plan]\label{def:plan}
A \emph{study plan} (or \emph{schedule}) of length $m$ is a sequence $\pi = (T_1,\dots,T_m)$ of pairwise disjoint subsets of $V$. Write $\taken(\pi) \eqdef \bigcup_{i=1}^{m} T_i$ and, for $1 \le i \le m$, $\pre_i(\pi) \eqdef \bigcup_{\ell < i} T_\ell$. The plan is \emph{feasible} for $C$ if
\begin{enumerate}[label=(F\arabic*),nosep]
  \item $\credits(T_i) \le \kappa$ for every $i$ \hfill(capacity);
  \item for every $i$ and every $v \in T_i$, $\ \pre_i(\pi) \models P(v)$ \hfill(prerequisite).
\end{enumerate}
The plan is \emph{complete} if $\credits\big(\taken(\pi) \cap S_g\big) \ge r_g$ for every $(S_g,r_g) \in R$. The \emph{plan language} of $C$ is
\[
  L(C) \eqdef \{\, \pi : \pi \text{ is feasible and complete for } C \,\},
\]
a language over the alphabet $2^{V}$.
\end{definition}

Since the terms of a plan are pairwise disjoint, at most $n$ of them are nonempty, and empty terms never help either objective below; we therefore assume without loss of generality that plans have no empty term, so that $m \le n$ and $L(C)$ is finite.

\begin{definition}[Objectives]\label{def:obj}
The \emph{time to degree} and the \emph{load} of $C$ are
\begin{align*}
  \ttd(C)  &\eqdef \min\{\, m : (T_1,\dots,T_m) \in L(C)\,\}, \\
  \load(C) &\eqdef \min\{\, \credits(\taken(\pi)) : \pi \in L(C)\,\},
\end{align*}
each equal to $+\infty$ when $L(C) = \emptyset$.
\end{definition}

\begin{definition}[Fragments]\label{def:fragments}
$C$ is \emph{conjunctive} if $\omega_\vee(C) \le 1$, that is, every clause is a single course. $C$ is \emph{mandatory} if $R = \{ (\{v\}, \credits(v)) : v \in M \}$ for some $M \subseteq V$, the \emph{target set}; completeness then means $\taken(\pi) \supseteq M$. $C$ has \emph{free capacity} if $\kappa = \infty$.
\end{definition}

The model deliberately retains only the three commitments a catalog states, and several features of real programs are outside it. Term-offering patterns, that a course is taught only in certain terms, are omitted, so a plan may schedule any course in any term; corequisites, courses that must be taken together, are treated as prerequisites satisfiable in an earlier term or dropped rather than modeled as same-term constraints; and minimum-grade conditions and per-term or per-degree credit floors and ceilings beyond $\kappa$ are not represented. Each of these is a monotone addition to the feasibility relation and would only enlarge the constraint set, so the hardness results below persist under any of them. Course repetition, by contrast, is a relaxation that enlarges the space of feasible plans; since repeating a course never shortens a plan or reduces its load, it too leaves the results intact. The tractability results are stated for the model as defined and would need re-examination were offering or same-term constraints added.

\section{Recognition is easy}\label{sec:recognition}

\begin{proposition}\label{prop:membership}
Given a curriculum $C$ and a plan $\pi$, deciding whether $\pi \in L(C)$ takes time linear in $|\pi| + \sum_v |P(v)| + \sum_g |S_g|$.
\end{proposition}

\begin{proof}
Process the terms in order, maintaining the running set $\pre_i(\pi)$. For (F1) sum the weights in $T_i$. For (F2), for each $v \in T_i$ test each clause $A_{v,j}$ for nonempty intersection with $\pre_i(\pi)$, which by Observation~\ref{fact:clause} is the exact satisfaction test. After the last term, evaluate completeness by summing $\credits(v)$ over $v \in \taken(\pi) \cap S_g$ for each group. Every course, clause, and group membership is inspected a constant number of times.
\end{proof}

The content of the theory is thus in the optimization problems, not in recognition. Both objectives' decision versions therefore lie in \textup{NP}: an optimal plan has at most $n$ nonempty terms, so it is a polynomial-size certificate that Proposition~\ref{prop:membership} verifies in polynomial time, and every hardness result below is accordingly a \emph{completeness} result.

\section{Time to degree without a capacity bound}\label{sec:timefree-sec}

We first treat free capacity. The purpose of this section is to prove that under free capacity the time objective is polynomial in full generality, so that neither disjunction nor electives contribute to its hardness.

\subsection{The conjunctive, mandatory base case}

\begin{definition}[Closure and chain length]\label{def:closure}
Let $C$ be conjunctive. For $X \subseteq V$ let $\cl(X)$ be the least superset of $X$ closed under prerequisites, that is, the least $Y \supseteq X$ such that $u \in Y$ whenever $u$ is a prerequisite of some $v \in Y$; this is well defined by well-foundedness. For $v \in V$ let $\ell(v)$ be the number of vertices on a longest directed path of $D(C)$ ending at $v$, so $\ell(v) = 1$ if $v$ has no prerequisite and $\ell(v) = 1 + \max\{\ell(u) : (u,v) \in E\}$ otherwise.
\end{definition}

\begin{theorem}[Flat-list characterization]\label{thm:flatlist}
Let $C$ be conjunctive, mandatory with target $M$, and have free capacity. Then
\[
  \ttd(C) \;=\; \max_{v \in \cl(M)} \ell(v),
\]
with the convention that the maximum over the empty set is $0$; both the value and an optimal plan are computable in time $O(n + |E|)$.
\end{theorem}

\begin{proof}
\emph{Lower bound.} Every course of $\cl(M)$ must be taken: $M$ is taken by completeness, and if $v$ is taken then, by (F2) and conjunctivity, each of its prerequisites lies in some earlier term and is therefore taken. Now let $v_1, v_2, \dots, v_t$ be a longest path in $D(C)$ ending at a vertex of $\cl(M)$, so $t = \max_{v\in\cl(M)} \ell(v)$ and every $v_s$ lies in $\cl(M)$. If $v_s \in T_{i_s}$ in a feasible plan, then $v_s$ is a prerequisite of $v_{s+1}$, so (F2) forces $i_s < i_{s+1}$. Hence $i_1 < i_2 < \cdots < i_t$ and the plan has at least $t$ terms.

\emph{Upper bound.} Define $T_i \eqdef \{ v \in \cl(M) : \ell(v) = i \}$ for $1 \le i \le t$. These sets are disjoint and cover $\cl(M)$. Free capacity gives (F1). For (F2), if $v \in T_i$ then every prerequisite $u$ of $v$ satisfies $\ell(u) \le i-1$, hence $u \in \pre_i(\pi)$; as $C$ is conjunctive, satisfying every prerequisite in this way makes $\pre_i(\pi) \models P(v)$. The plan takes $\cl(M) \supseteq M$, so it is complete, and it has $t$ terms. The values $\ell(v)$ and the layering are produced by one topological pass, which is linear-time by Proposition~\ref{fact:topo}.
\end{proof}

\subsection{Disjunction does not raise the cost of time}

For general prerequisites the layering must choose, in each clause, which alternative to rely on. The next definition packages that choice into a level function, and the theorem after it shows the choice is made greedily and optimally.

\begin{definition}[AND/OR level]\label{def:level}
For a well-founded curriculum $C$ define $h : V \to \Zpos$ in topological order of $D(C)$ by
\[
  h(v) \;\eqdef\; 1 + \max_{1 \le j \le k_v}\ \min_{u \in A_{v,j}} h(u),
\]
with the convention that the maximum over the empty set of clauses ($k_v = 0$) is $0$, so $h(v) = 1$ when $P(v) = \top$.
\end{definition}

The recurrence is well defined because every $u \in A_{v,j}$ is an in-neighbour of $v$ in $D(C)$ and therefore precedes $v$ in topological order.

\begin{lemma}[Level is a lower bound on completion term]\label{lem:levellower}
Let $\pi$ be any feasible plan for $C$ and let $v \in T_i$. Then $i \ge h(v)$.
\end{lemma}

\begin{proof}
Induction on $i$. If $P(v) = \top$ then $i \ge 1 = h(v)$. Otherwise fix a clause $j$. By (F2), $A_{v,j} \cap \pre_i(\pi) \neq \emptyset$, so some $u_j \in A_{v,j}$ lies in a term $T_{i_j}$ with $i_j < i$. By the induction hypothesis $i_j \ge h(u_j) \ge \min_{u \in A_{v,j}} h(u)$. As this holds for every clause $j$, $\ i \ge 1 + \max_j i_j \ge 1 + \max_j \min_{u \in A_{v,j}} h(u) = h(v)$.
\end{proof}

\begin{theorem}[Disjunctive, mandatory, free capacity]\label{thm:disjfree}
Let $C$ be a well-founded curriculum with arbitrary prerequisites, mandatory with target $M$, and free capacity. Then $\ttd(C) = \max_{v \in M} h(v)$, with the convention $\max\emptyset = 0$, computable in time $O\big(n + \sum_v |P(v)|\big)$.
\end{theorem}

\begin{proof}
\emph{Lower bound.} Each $v \in M$ is taken (completeness), say in term $i_v$; Lemma~\ref{lem:levellower} gives $i_v \ge h(v)$, so the plan has at least $\max_{v \in M} h(v)$ terms.

\emph{Upper bound.} For each $v$ with $k_v \ge 1$ and each clause $j$ pick a witness $u^\ast_{v,j} \in \arg\min_{u \in A_{v,j}} h(u)$. Let $Y$ be the least set containing $M$ and closed under the rule ``if $v \in Y$ and $k_v \ge 1$ then $u^\ast_{v,j} \in Y$ for all $j$''. Schedule $T_i \eqdef \{ v \in Y : h(v) = i \}$ for $1 \le i \le \max_{v\in M} h(v)$. Free capacity gives (F1). For (F2): if $v \in T_i$ and $k_v \ge 1$, then for each clause $j$ the witness $u^\ast_{v,j}$ has $h(u^\ast_{v,j}) = \min_{u\in A_{v,j}} h(u) \le h(v) - 1 = i-1$, so $u^\ast_{v,j} \in \pre_i(\pi)$ and clause $j$ is satisfied; thus $\pre_i(\pi) \models P(v)$. The schedule takes $Y \supseteq M$, so it is complete, and its length is $\max_{v \in M} h(v)$. The values $h(v)$ are computed by one topological pass, each clause inspected once.
\end{proof}

\begin{corollary}\label{cor:widthirrelevant}
The time to degree of a mandatory, free-capacity curriculum is polynomial-time computable for every value of the disjunctive width $\omega_\vee(C)$; disjunctive width does not affect the tractability of the time objective.
\end{corollary}

\subsection{Electives do not raise the cost of time}

We now remove the mandatory restriction and allow arbitrary requirement groups, still under free capacity.

\begin{theorem}[Time is polynomial without a capacity bound]\label{thm:timefree}
Let $C$ be a well-founded curriculum with arbitrary prerequisites, arbitrary requirement groups $R$, and free capacity. Then
\[
  \ttd(C) \;=\; \min\Big\{\, B \in \Znn \ :\ \credits\big(\{ v \in S_g : h(v) \le B \}\big) \ge r_g \ \text{ for every } (S_g,r_g) \in R \,\Big\},
\]
with the minimum over the empty set understood as $+\infty$. Consequently $\ttd(C)$ is computable in polynomial time.
\end{theorem}

\begin{proof}
Write $V_{\le B} \eqdef \{ v \in V : h(v) \le B \}$. Note first that $V_{\le B}$ is closed under the witness rule of Theorem~\ref{thm:disjfree}: if $h(v) \le B$ then each witness $u^\ast_{v,j}$ has $h(u^\ast_{v,j}) < h(v) \le B$.

\emph{Every complete plan of length $B$ uses only $V_{\le B}$.} If $\pi \in L(C)$ has length $B$, then every taken course lies in some term $i \le B$, and by Lemma~\ref{lem:levellower} has $h(v) \le i \le B$; hence $\taken(\pi) \subseteq V_{\le B}$. Completeness then requires $\credits(\taken(\pi) \cap S_g) \ge r_g$, and since $\taken(\pi) \cap S_g \subseteq \{ v \in S_g : h(v) \le B\}$, the displayed inequality holds for that $B$. Thus $\ttd(C)$ is at least the stated minimum.

\emph{The threshold $B$ is achievable.} Let $B$ satisfy the displayed condition. Schedule $T_i \eqdef \{ v \in V_{\le B} : h(v) = i\}$ for $1 \le i \le B$. Exactly as in the proof of Theorem~\ref{thm:disjfree}, the closure property of $V_{\le B}$ under witnesses makes this schedule satisfy (F1) and (F2); it takes all of $V_{\le B}$, and by the displayed condition $\credits(V_{\le B} \cap S_g) \ge r_g$ for every group, so it is complete. Its length is $B$. Hence $\ttd(C) \le B$ for every feasible $B$, and equality with the minimum follows.

The values $h(v)$ are computed once; the candidate thresholds $B$ are the at most $n$ distinct values of $h$, and each is tested in linear time.
\end{proof}

\begin{remark}
Theorem~\ref{thm:timefree} is the strongest tractability statement of the paper: with no capacity bound, the time objective is polynomial for \emph{arbitrary} disjunctive prerequisites and \emph{arbitrary} credit-threshold electives. Whatever makes time to degree hard, it is not the logical structure of prerequisites and not the combinatorics of electives.
\end{remark}

\section{Time to degree with a capacity bound}\label{sec:capacity}

The single remaining ingredient, the per-term capacity, is by itself sufficient for intractability.

\begin{problem}[\textsc{Time-to-Degree}]\label{prob:ttd}
Given a curriculum $C$ and an integer $B$, decide whether $\ttd(C) \le B$.
\end{problem}

\begin{theorem}[Capacity makes time hard]\label{thm:capacityhard}
\textsc{Time-to-Degree} is strongly NP-hard, already for instances that are conjunctive and mandatory and in which every course has the empty prerequisite $\top$.
\end{theorem}

\begin{proof}
Reduction from \textsc{Bin Packing} (Definition~\ref{def:binpacking}), which is strongly NP-hard by Proposition~\ref{fact:binpacking}. Given item sizes $s_1,\dots,s_N \in \Zpos$, a bin capacity $\kappa$, and a bound $B$, build $C$ with a course $v_i$ of weight $\credits(v_i) = s_i$ for each item, all prerequisites equal to $\top$, capacity $\kappa$, and mandatory target $M = V$ (equivalently $R = \{(\{v_i\}, s_i)\}_i$). A plan is any assignment of the courses to terms with (F1) $\credits(T_i) \le \kappa$; prerequisites are vacuous, and completeness forces all courses to be taken. Thus feasible complete plans of length $B$ correspond exactly to packings of the items into $B$ bins of capacity $\kappa$, and $\ttd(C) \le B$ iff the packing instance is a yes-instance. The construction is polynomial and preserves the numerical values, so by Propositions~\ref{fact:transfer} and~\ref{fact:binpacking} the strong NP-hardness transfers.
\end{proof}

\begin{remark}[A sharper reading through scheduling]\label{rem:scheduling}
When all credits are equal to a common value $\bar{w}$, so that (F1) caps the \emph{number} of courses per term at $c \eqdef \lfloor \kappa / \bar{w} \rfloor$, \textsc{Time-to-Degree} on a conjunctive mandatory instance is exactly the problem of scheduling unit-length jobs with precedence constraints on $c$ identical machines to minimize the makespan, the classical $P \mid \mathrm{prec},\, p_j = 1 \mid C_{\max}$. By Proposition~\ref{fact:scheduling} that problem is NP-hard when $c$ is part of the input, so the capacity bound is intractable even in the unit-credit case, and its status for each fixed $c \ge 3$ is a long-standing open question that the curricular formulation inherits verbatim.
\end{remark}

\begin{remark}[The hardness in realistic parameter ranges]\label{rem:realistic}
Theorem~\ref{thm:capacityhard} draws its hardness from instances with unbounded credit weights, whereas an actual catalog is milder on both counts: credit values are small integers and the per-term capacity is a fixed ceiling of a few courses. In that regime the prerequisite-free case is polynomial, since \textsc{Bin Packing} with a constant number of distinct item sizes and constant capacity is solvable in polynomial time, and the unit-credit case with prerequisites is exactly $P \mid \mathrm{prec},\, p_j = 1 \mid C_{\max}$ with a \emph{fixed} number of machines, whose complexity for each fixed $c \ge 3$ is the long-standing open problem of Proposition~\ref{fact:scheduling}. The strong NP-hardness of Theorem~\ref{thm:capacityhard} is therefore a statement about the model in full generality; for the parameter ranges a real curriculum occupies, the exact complexity of the time objective is open rather than settled. This does not weaken the paper's organizing claim, which is one of \emph{location} rather than of worst-case intractability: capacity, and capacity alone, is the commitment on which the tractability of the time objective turns, while disjunction and electives leave it untouched.
\end{remark}

\section{Load, and the duality of hardness}\label{sec:load}

We turn to the load objective and show that its hardness has a different and complementary source.

\begin{problem}[\textsc{Min-Load}]\label{prob:load}
Given a curriculum $C$ and an integer $L$, decide whether $\load(C) \le L$.
\end{problem}

\begin{proposition}[Conjunctive load is easy]\label{prop:loadconj}
If $C$ is conjunctive and mandatory with target $M$, and $\kappa \ge \max_{v \in \cl(M)} \credits(v)$ (in particular if $\kappa = \infty$), then $\load(C) = \credits(\cl(M))$, computable in time $O(n + |E|)$. For smaller $\kappa$ some course of $\cl(M)$ cannot be scheduled and $\load(C) = +\infty$.
\end{proposition}

\begin{proof}
By the lower-bound argument of Theorem~\ref{thm:flatlist}, every feasible complete plan takes at least $\cl(M)$, so $\load(C) \ge \credits(\cl(M))$. Conversely $\cl(M)$ can be scheduled: place its courses in prerequisite order, one nonempty term at a time if necessary, which respects (F1) for any $\kappa \ge \max_v \credits(v)$ and (F2) by construction; if some single course exceeds $\kappa$ the instance is infeasible and $\load = +\infty$ consistently. Hence $\load(C) = \credits(\cl(M))$. The closure is computed by one graph traversal.
\end{proof}

\begin{theorem}[Disjunction makes load hard]\label{thm:loadhard}
\textsc{Min-Load} is NP-hard, already for instances that are mandatory, have free capacity, and have disjunctive width two.
\end{theorem}

\begin{proof}
Reduction from \textsc{Set Cover} (Definition~\ref{def:setcover}), which is NP-hard by Proposition~\ref{fact:setcover}. Given a universe $U = \{e_1,\dots,e_q\}$, a family $F_1,\dots,F_s \subseteq U$, and an integer $t$, construct $C$ as follows. For each set $F_i$ introduce a course $x_i$ with $\credits(x_i) = 1$ and $P(x_i) = \top$. For each element $e \in U$ introduce a course $y_e$ with $\credits(y_e) = q+1$ and the single disjunctive clause
\[
  P(y_e) \;=\; \bigvee_{i \,:\, e \in F_i} x_i .
\]
Let the target set be $M = \{ y_e : e \in U\}$ (all element courses mandatory), and set $\kappa = \infty$. Any feasible complete plan takes every $y_e$, contributing fixed load $q(q+1)$, and for each $e$ must take at least one $x_i$ with $e \in F_i$ to satisfy $P(y_e)$; the set of taken $x_i$ is therefore a cover of $U$, of load equal to its cardinality. Conversely any cover yields such a plan under free capacity. Hence
\[
  \load(C) \;=\; q(q+1) + \min\{\, |\mathcal{K}| : \mathcal{K} \text{ covers } U \,\},
\]
and $\load(C) \le q(q+1) + t$ iff the \textsc{Set Cover} instance has a cover of size $t$. The reduction is polynomial. Finally, restricting to \textsc{Vertex Cover}, the special case in which every element lies in exactly two sets and which remains NP-hard \citep{karp1972}, makes every clause $P(y_e)$ a disjunction of exactly two courses, so the hardness already holds at disjunctive width two; bounded disjunctive width therefore does not make the load objective tractable.
\end{proof}

The disjunction reduction uses only polynomially bounded weights and so yields \emph{strong} NP-hardness. Load has a second and independent source of hardness, the elective choice, which is hard even with no disjunction at all.

\begin{theorem}[Electives make load hard]\label{thm:loadelective}
\textsc{Min-Load} is NP-hard, already for instances that are conjunctive, prerequisite-free, and have free capacity and a single requirement group.
\end{theorem}

\begin{proof}
Reduction from \textsc{Subset Sum}, which is NP-complete \citep{gj1979}: given positive integers $a_1,\dots,a_N$ and a target $b$, decide whether some subset sums to exactly $b$. Build $C$ with courses $z_1,\dots,z_N$ of weights $\credits(z_i) = a_i$, all prerequisites $\top$, the single requirement group $(Z, b)$ with $Z = \{z_1,\dots,z_N\}$, and $\kappa = \infty$. Because prerequisites are vacuous, any subset of $Z$ is the taken set of a feasible plan under free capacity, and such a plan is complete iff its taken set weighs at least $b$; hence $\load(C) = \min\{ \credits(X) : X \subseteq Z,\ \credits(X) \ge b\}$. As the weights are integers, $\load(C) \le b$ holds iff some subset weighs exactly $b$, that is iff the \textsc{Subset Sum} instance is a yes-instance. The reduction is polynomial.
\end{proof}

The route through \textsc{Subset Sum} is only weakly NP-hard: the prerequisite-free single-group fragment is solvable in pseudo-polynomial time by the standard subset-sum dynamic program. Elective choice is nevertheless a source of \emph{strong} hardness as soon as a course may count toward more than one requirement, which is the norm in real catalogs, where a single course often satisfies several overlapping requirement groups.

\begin{theorem}[Overlapping electives make load strongly hard]\label{thm:loadover}
\textsc{Min-Load} is strongly NP-hard, already for instances that are conjunctive, prerequisite-free, have free capacity, and use only unit credits, once requirement groups are permitted to overlap.
\end{theorem}

\begin{proof}
Reduction from \textsc{Set Cover} (Definition~\ref{def:setcover}), NP-hard by Proposition~\ref{fact:setcover}. Given a universe $U = \{e_1,\dots,e_q\}$, a family $F_1,\dots,F_s \subseteq U$, and an integer $t$, introduce for each set $F_i$ a course $x_i$ with $\credits(x_i) = 1$ and $P(x_i) = \top$, and for each element $e \in U$ a requirement group $(S_e, 1)$ with $S_e = \{ x_i : e \in F_i \}$. Set $\kappa = \infty$. Since every prerequisite is $\top$, any subset of $\{x_1,\dots,x_s\}$ is the taken set of a feasible plan under free capacity, and such a plan is complete iff, for every $e$, it takes at least one $x_i$ with $e \in F_i$, that is, iff the taken courses form a cover of $U$. As every credit is one, the load of a plan equals the number of courses it takes, so $\load(C) = \min\{\, |\mathcal{K}| : \mathcal{K} \text{ covers } U\,\}$ and $\load(C) \le t$ iff $U$ has a cover of size at most $t$. The reduction is polynomial and every integer in the constructed instance equals $1$, so the hardness is strong.
\end{proof}

Elective choice therefore matches disjunction in strength once groups overlap, both giving strongly NP-hard load; only the restricted case of a single requirement group, equivalently pairwise-disjoint groups, is confined to weak, pseudo-polynomial hardness. The two commitments remain genuinely independent, since each is hard in the total absence of the other, disjunction with mandatory (hence non-overlapping) requirements by Theorem~\ref{thm:loadhard} and overlapping electives with purely conjunctive prerequisites by Theorem~\ref{thm:loadover}. The load objective under electives is a covering analogue of the partially ordered knapsack problem, whose hardness and approximability were studied by \citet{kolliopoulos2007}.

\begin{theorem}[Duality of hardness]\label{thm:duality}
For the two planning objectives, the sources of intractability are disjoint:
\begin{center}
\setlength{\tabcolsep}{4.5pt}\begin{tabular}{lll}
\hline
 & \emph{free capacity} $(\kappa = \infty)$ & \emph{bounded capacity} \\
\hline
\textsc{Time-to-Degree} & polynomial \emph{(Thm.~\ref{thm:timefree})} & strongly NP-complete \emph{(Thm.~\ref{thm:capacityhard})} \\
\textsc{Min-Load} & strongly NP-complete \emph{(Thms.~\ref{thm:loadhard},~\ref{thm:loadover})} & strongly NP-complete \\
\hline
\end{tabular}
\end{center}
Time to degree is hard only through the capacity bound and is insensitive to both disjunction and elective choice (Theorem~\ref{thm:timefree}). The complexity of the load objective does not depend on the capacity bound; load is hard through either of two independent commitments, disjunction (Theorem~\ref{thm:loadhard}) or overlapping electives (Theorem~\ref{thm:loadover}), and is polynomial on the conjunctive, mandatory fragment (Proposition~\ref{prop:loadconj}). The sources of hardness are thus disjoint: capacity governs the time objective alone, while disjunction and electives govern the load objective alone.
\end{theorem}

\begin{proof}
The time cells are Theorems~\ref{thm:timefree} and~\ref{thm:capacityhard}. The load cells are Theorems~\ref{thm:loadhard} and~\ref{thm:loadover}, whose disjunction and overlapping-elective reductions both use polynomially bounded weights and so give strong hardness, together with the observation that both reductions are unaffected by imposing a finite $\kappa$ no smaller than the largest single course weight, since their courses may be scheduled one per term. Membership in \textup{NP} holds for both problems, as noted after Proposition~\ref{prop:membership}, so the hard cells are completeness results. Polynomiality on the conjunctive mandatory fragment is Proposition~\ref{prop:loadconj}, and under free capacity time is polynomial by Theorem~\ref{thm:timefree}.
\end{proof}

\section{The connection to the curricular-complexity metric}\label{sec:metric-sec}

We close the formal development by relating the model to the structural complexity score of curricular analytics, which is defined on the prerequisite graph under the all-conjunctive reading.

\begin{definition}[Conjunctive collapse]\label{def:collapse}
The \emph{conjunctive collapse} of a curriculum $C$ is the conjunctive curriculum $C^{\wedge}$ obtained by replacing each clause $\bigvee_{u \in A_{v,j}} u$ of every prerequisite by the conjunction $\bigwedge_{u \in A_{v,j}} u$, that is, by treating every listed alternative as required. All other data are unchanged.
\end{definition}

\begin{proposition}[The metric is an exact-on-the-fragment upper bound]\label{prop:metric}
For every well-founded curriculum $C$ and every course $v$,
\[
  h(v) \;\le\; \ell_{C^{\wedge}}(v),
\]
where $\ell_{C^{\wedge}}$ is the chain-length function of Definition~\ref{def:closure} computed in the collapse $C^{\wedge}$. Consequently, for mandatory free-capacity curricula,
\[
  \ttd(C) \;\le\; \ttd(C^{\wedge}) \;=\; \max_{v \in \cl(M)} \ell_{C^{\wedge}}(v),
\]
with equality throughout when $C$ is conjunctive. The nonnegative quantity $\ttd(C^{\wedge}) - \ttd(C)$ is the \emph{disjunctive slack}.
\end{proposition}

\begin{proof}
Induction on topological order. If $P(v) = \top$ then $h(v) = 1 = \ell_{C^{\wedge}}(v)$. Otherwise, in $C^{\wedge}$ every alternative of every clause is a prerequisite arc, so $\ell_{C^{\wedge}}(v) = 1 + \max_j \max_{u \in A_{v,j}} \ell_{C^{\wedge}}(u)$. Using the induction hypothesis $h(u) \le \ell_{C^{\wedge}}(u)$ and the elementary inequality $\min \le \max$ on each clause,
\begin{align*}
  h(v) &= 1 + \max_j \min_{u \in A_{v,j}} h(u) \\
       &\le 1 + \max_j \max_{u \in A_{v,j}} \ell_{C^{\wedge}}(u) \\
       &= \ell_{C^{\wedge}}(v).
\end{align*}
Taking the maximum over $v \in M$ and applying Theorems~\ref{thm:disjfree} and~\ref{thm:flatlist} gives the displayed relation between the two times to degree. When $C$ is conjunctive every clause is a singleton, $\min$ and $\max$ coincide, and the inequality is an equality.
\end{proof}

\begin{remark}[Interpretation]
The maximal chain length of the conjunctive collapse is exactly the maximum \emph{delay factor} of curricular analytics, one of the two components (together with the blocking factor) whose course-wise sum constitutes the aggregate curricular-complexity score \citep{slim2014,heileman2018}. Proposition~\ref{prop:metric} therefore concerns that delay-factor component specifically, and states that it is a polynomially computable \emph{upper bound} on the true time to degree, tight precisely on the conjunctive curricula for which the graph abstraction is faithful, and loose elsewhere by the disjunctive slack. The slack is strictly positive only when some clause offers alternatives of differing level, whether in the course's own prerequisite or, through propagation, in that of an ancestor; this is the structural condition under which the all-conjunctive reading overstates rigidity, and it is the formal counterpart of the empirical correction reported across institutions by our companion study \citep{turaev2026}. Figure~\ref{fig:worked} works a small curriculum through both readings, and it exhibits a second effect that the empirical section will quantify: the slack that originates at a disjunctive course propagates downward along the conjunctive chain above it.
\end{remark}

\begin{figure}[htbp]
\centering
\Description{Two five-level layered diagrams of the same nine courses shown side by side. The left diagram gives the true AND/OR levels; the right diagram gives the conjunctive collapse, in which three courses each move up by one level. The comparison shows that treating a non-binding alternative as required raises the level of a course and of its purely conjunctive descendants, while the deepest all-conjunctive chain, and therefore the time to degree, is left unchanged.}
\includegraphics[width=\linewidth]{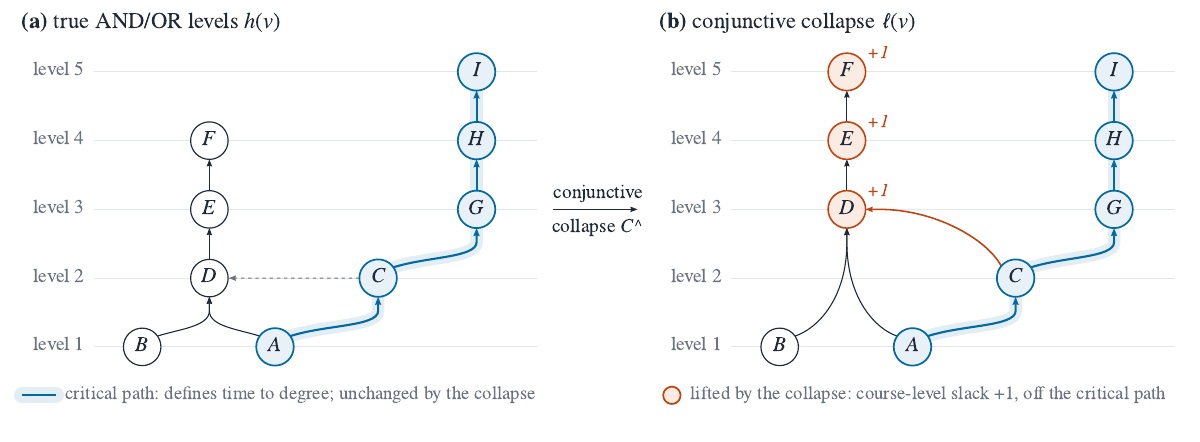}
\caption{A nine-course example separating the two readings; every course other than $D$ requires the single course beneath it, and $P(D)=(A\vee C)\wedge B$. (a) The true AND/OR levels $h(v)$: the clause of $D$ is met by $A$ alone, so the alternative $C$, drawn dashed, is non-binding. (b) The conjunctive collapse $C^{\wedge}$, levels $\ell(v)$: every listed alternative is treated as required, forcing $C$ before $D$ and lifting $D$, $E$, and $F$ by one level each (terracotta, marked $+1$). The lift at $E$ and $F$ is inherited along the conjunctive chain above $D$, which is how the course-level slack $s(v)=\ell_{C^{\wedge}}(v)-h(v)$ propagates. The deepest chain $A\to C\to G\to H\to I$ is conjunctive throughout, so the collapse leaves it unchanged and the curriculum-level slack $\ttd(C^{\wedge})-\ttd(C)=5-5$ is zero}
\label{fig:worked}
\end{figure}

\section{The empirical boundary across the corpus}\label{sec:empirical}

The value of the model rests not only on its theorems but on where real curricula fall relative to the boundaries they draw. We instantiate the parameters of the model on a corpus of twenty-two universities of the United Arab Emirates, taking for each institution its most recent catalog edition, a cross-section of $6{,}422$ courses that carry a parsed prerequisite expression once corequisite relations are set aside (see the data cautions). The per-institution figures below are reported for the seventeen institutions with at least twenty such courses; the three smaller institutions are included in the totals but omitted from Table~\ref{tab:byinst} and the figures.

\subsection{Real curricula inhabit the tractable fragment}

Of the $6{,}422$ prerequisite-bearing courses in the cross-section, $88.0$ percent are purely conjunctive and only $12.0$ percent carry a disjunction, so the expressive feature that governs the load objective and the metric correction is, at the level of the individual course, the exception. The two width parameters that the theory identifies are small and concentrated, as Table~\ref{tab:widths} shows: the conjunctive width has median one and ninetieth percentile two, and the disjunctive width, measured only on the clauses that actually branch, has median two and ninetieth percentile four. By Corollary~\ref{cor:widthirrelevant} the disjunctive width never bears on the tractability of the time objective. It does not rescue the load objective either: by Theorem~\ref{thm:loadhard} load is NP-hard already at disjunctive width two, so the small observed widths do not place these instances in a tractable class.

\begin{table}[!htb]
\centering
\small
\caption{Distribution of the two width parameters over the $6{,}422$ prerequisite-bearing courses of the cross-section. Conjunctive width is the number of clauses per course; disjunctive width is the number of alternatives per branching clause. The final bin collects all larger values}
\label{tab:widths}
\begin{tabular}{lrrrrrr}
\toprule
value & $1$ & $2$ & $3$ & $4$ & $5$ & $\ge 6$ \\
\midrule
courses by conjunctive width & $4903$ & $1122$ & $203$ & $94$ & $22$ & $78$ \\
branching clauses by disjunctive width & --- & $701$ & $121$ & $64$ & $36$ & $35$ \\
\bottomrule
\end{tabular}
\end{table}

\subsection{Disjunction is concentrated, not spread}

The disjunction that drives both the load hardness and the metric correction is not distributed evenly across the system. The blue bars of Figure~\ref{fig:shares} rank the institutions by the share of their prerequisite-bearing courses that carry an alternative, and the distribution is sharply bimodal: two institutions, Khalifa University at $45.0$ percent and the American University of Sharjah at $34.3$ percent, stand far above the rest, a middle group sits near or below ten percent, and several institutions publish essentially no alternatives at all. These two institutions are precisely the ones at which our companion measurement \citep{turaev2026} finds the largest disjunctive correction to curricular complexity, so the theoretical boundary and the empirical correction pick out the same places from opposite directions. The terracotta bars of the same figure report the share of courses on which the all-conjunctive reading is actually wrong, a quantity we take up next.

\begin{figure}[htbp]
\centering
\Description{Paired horizontal bars for seventeen UAE institutions, sorted by disjunctive share, comparing the share of prerequisite-bearing courses that carry a disjunction (solid bars) with the share whose course-level slack is positive (hatched bars). An inset histogram gives the distribution of slack magnitude, from one to five, over the 573 slack-positive courses. Khalifa University and the American University of Sharjah lead on disjunctive share, while at the United Arab Emirates University and the American University in Dubai the slack-positive share exceeds the disjunctive share.}
\includegraphics[width=0.82\linewidth]{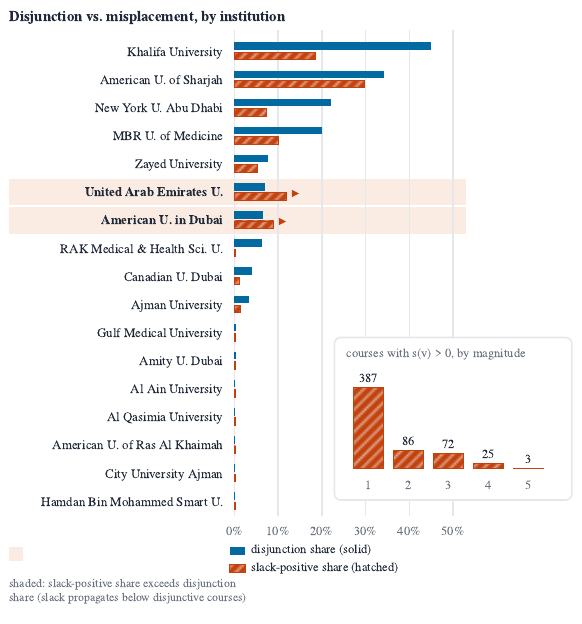}
\caption{Where the all-conjunctive reading misplaces courses. For each of the seventeen institutions with at least twenty prerequisite-bearing courses, taking the latest catalog edition and excluding corequisites, the share of prerequisite-bearing courses that carry a disjunction (blue, solid) beside the share whose course-level slack $s(v)=\ell_{C^{\wedge}}(v)-h(v)$ is positive (terracotta, hatched), sorted by the disjunctive share; the inset gives the distribution of slack magnitude over the $573$ courses with $s(v)>0$. At most institutions the two shares track one another, but at the United Arab Emirates University and the American University in Dubai the slack-positive share exceeds the disjunctive share, because slack propagates from disjunctive courses down the purely conjunctive chains beneath them}
\label{fig:shares}
\end{figure}

\subsection{The slack is real, and shallow chains leave capacity to bind}

Proposition~\ref{prop:metric} concerns the disjunctive slack at the level of a whole curriculum, the gap $\ttd(C^\wedge) - \ttd(C)$ between the collapse depth and the true time to degree. It is worth separating that curriculum-level quantity from the \emph{course-level} slack $s(v) \eqdef \ell_{C^\wedge}(v) - h(v)$, the amount by which the all-conjunctive reading overstates the level of an individual course. The two behave differently in the data, and the difference is itself informative. The course-level slack is positive on $573$ courses, $8.9$ percent of the cross-section, with a mean of $1.55$ terms among those and a maximum of five, so the all-conjunctive reading inflates the depth of one course in eleven, by more than a term and a half on average. The curriculum-level slack, by contrast, is exactly zero at every institution: the longest chain and the maximum level coincide in every case, because the deepest chain of each program runs through mandatory conjunctive courses that no alternative shortens. Disjunction, in other words, redistributes the depths of individual courses without shortening any program's critical path. The overstatement is modest in magnitude as well as in reach, as the inset of Figure~\ref{fig:shares} shows: of the $573$ slack-positive courses, $387$ are overstated by a single level and the count falls away steeply thereafter, to $86$ at two levels, $72$ at three, $25$ at four, and only three at five.

The terracotta bars of Figure~\ref{fig:shares} record, for each institution, the share of prerequisite-bearing courses on which the slack is positive, set against the disjunctive share in blue. The slack-positive share is positive only where alternatives of differing level occur, as Proposition~\ref{prop:metric} requires, and its largest values fall at the high-share institutions, though the relation to the disjunctive share is not monotone: the American University of Sharjah, at a disjunctive share below Khalifa's, carries the larger slack-positive share, $29.7$ against $18.6$ percent. Two institutions invert the ordering outright, in that their slack-positive share \emph{exceeds} their disjunctive share, the United Arab Emirates University with $11.9$ against $7.1$ percent and the American University in Dubai with $9.0$ against $6.6$ percent. The inversion has a structural cause worth naming. A purely conjunctive course carries no alternative of its own, so under Proposition~\ref{prop:metric} it can acquire positive slack only from a prerequisite that already carries it, because if every prerequisite of a conjunctive course had zero slack then its level would coincide under the two readings; the slack therefore \emph{originates} at a disjunctive course and \emph{propagates} downward along the conjunctive chains beneath it, exactly as Figure~\ref{fig:worked} illustrates at the courses $E$ and $F$. Across the corpus this propagation accounts for nearly half of the affected courses: of the $573$ slack-positive courses, $299$ are disjunctive in their own right and the remaining $274$ are purely conjunctive descendants that inherit their slack. At the two institutions that invert the ordering the propagated courses are the majority, $112$ of $160$ at the United Arab Emirates University and $30$ of $45$ at the American University in Dubai, which is why a curriculum with few alternatives can still present many courses whose depth the all-conjunctive reading overstates.

Table~\ref{tab:byinst} gives the institution-level detail, including the longest prerequisite chain, which is the depth that would bound time to degree were capacity not to bind, and which equals the maximum level in every row. That coincidence is the empirical counterpart of the hardness duality of Theorem~\ref{thm:duality}. The longest chain has a median across institutions of about eight courses and nowhere exceeds ten. Since a chain of length $\ell$ forces at least $\ell$ terms, a depth of eight to ten sits close to the nominal duration of a bachelor's degree, so prerequisite depth does not by itself extend a program beyond the schedule a student would follow in any case; it is the per-term capacity, the number of courses a term can hold, rather than the logical structure of prerequisites, that governs whether the degree can be completed on time. This reading is consistent with the theory, disjunction being tractable for the time objective by Theorem~\ref{thm:timefree} and the capacity bound being the commitment under which that objective turns hard (Theorem~\ref{thm:capacityhard}), with the practical qualification of Remark~\ref{rem:realistic}. The corpus thus locates real curricula on the tractable side of every boundary but one, and identifies that one, capacity, as the constraint a planning tool must actually contend with.

\begin{table}[t]
\centering
\footnotesize
\setlength{\tabcolsep}{3pt}
\caption{Institution-level cross-section (latest edition, seventeen institutions with at least twenty prerequisite-bearing courses), ordered by disjunctive share. ``Longest chain'' is the number of courses on the longest prerequisite chain, which equals the maximum level in every row, so the curriculum-level disjunctive slack is zero throughout; ``slack $>0$'' counts courses whose conjunctive-collapse level exceeds their true level; ``mean slack'' averages that course-level slack over those courses}
\label{tab:byinst}
\begin{tabular}{lrrrrr}
\toprule
institution (latest edition) & CNF courses & disj.\ \% & longest chain & slack ${>}0$ & mean slack \\
\midrule
Khalifa University & $625$ & $45.0$ & $8$ & $116$ & $1.39$ \\
American University of Sharjah & $671$ & $34.3$ & $8$ & $199$ & $2.16$ \\
New York University Abu Dhabi & $314$ & $22.0$ & $9$ & $23$ & $1.52$ \\
Mohammed Bin Rashid U.\ (Medicine) & $20$ & $20.0$ & $4$ & $2$ & $1.00$ \\
Zayed University & $420$ & $7.6$ & $10$ & $22$ & $1.00$ \\
United Arab Emirates University & $1340$ & $7.1$ & $8$ & $160$ & $1.10$ \\
American University in Dubai & $499$ & $6.6$ & $10$ & $45$ & $1.22$ \\
RAK Medical and Health Sciences U. & $64$ & $6.2$ & $8$ & $0$ & $0.00$ \\
Canadian University Dubai & $320$ & $4.1$ & $10$ & $4$ & $2.00$ \\
Ajman University & $152$ & $3.3$ & $7$ & $2$ & $1.00$ \\
Gulf Medical University & $254$ & $0.4$ & $9$ & $0$ & $0.00$ \\
Amity University Dubai & $367$ & $0.3$ & $8$ & $0$ & $0.00$ \\
Al Ain University & $591$ & $0.0$ & $7$ & $0$ & $0.00$ \\
Al Qasimia University & $125$ & $0.0$ & $5$ & $0$ & $0.00$ \\
American University of Ras Al Khaimah & $136$ & $0.0$ & $6$ & $0$ & $0.00$ \\
City University Ajman & $415$ & $0.0$ & $7$ & $0$ & $0.00$ \\
Hamdan Bin Mohammed Smart University & $82$ & $0.0$ & $7$ & $0$ & $0.00$ \\
\bottomrule
\end{tabular}
\end{table}

\subsection{Data cautions}

Several features of the corpus bear on the figures. Corequisite relations, the $185$ edges marked \emph{concurrent} in the cross-section, are not prerequisites, since a corequisite may be taken in the same term; they are therefore excluded from the prerequisite structure, which drops from the prerequisite-bearing set the $74$ courses whose only listed relations were corequisite. Counting them as prerequisites, as an earlier pass did, would inflate chain depth and level. Three further confounds are controlled. First, in several institution-editions the parsed prerequisite relation contains a directed cycle, a parsing artifact arising chiefly from cross-listed or renumbered courses; to keep the computation deterministic and order-independent, the depth and level computations remove every prerequisite arc whose two endpoints fall in the same strongly connected component of the support digraph, which can only understate depth, so the reported chains are lower bounds. Second, $687$ courses in the full corpus, none of them in the latest-edition cross-section analysed here, carry the corpus annotation \emph{degraded-flat-list}, indicating a source that rendered disjunctions as flat conjunctive lists; because the annotation touches no course in the cross-section, the reported shares and widths are unaffected by it, and across the full corpus such courses can only lose disjunction in transcription, so the disjunctive shares are, if anything, conservative. As a consistency check, the disjunctive share of this latest-edition cross-section, $12.0$ percent, agrees with the $10.7$ percent measured over all editions of all institutions, the small difference reflecting the growth of expressive notation in recent catalogs.

\section{Two comparison problems: subsumption and consensus}\label{sec:compare}

The plan-language formulation makes two further questions precise. We settle the complexity of each. The first compares two curricula; the second reconciles several.

\subsection{Completions, and why na\"ive transfer is vacuous}

\begin{definition}[Completion set]\label{def:comp}
The \emph{completion set} of a curriculum $C$ is $\mathrm{Comp}(C) \eqdef \{\, \taken(\pi) : \pi \in L(C)\,\} \subseteq 2^{V}$.
\end{definition}

\begin{lemma}[Characterization of completions]\label{lem:comp}
Let $C$ be well-founded with free capacity. Then $X \in \mathrm{Comp}(C)$ if and only if
\begin{enumerate}[label=(\alph*),nosep]
  \item $X \models P(v)$ for every $v \in X$ \hfill(prerequisite-closed), and
  \item $\credits(X \cap S_g) \ge r_g$ for every group $(S_g,r_g) \in R$ \hfill(requirement-satisfying).
\end{enumerate}
\end{lemma}

\begin{proof}
($\Rightarrow$) If $X = \taken(\pi)$ for a feasible complete plan, then for $v \in X$, say $v \in T_i$, (F2) gives $\pre_i(\pi) \models P(v)$, and $\pre_i(\pi) \subseteq X$, so by Observation~\ref{fact:monotone} $X \models P(v)$; completeness is exactly (b).
($\Leftarrow$) Suppose $X$ satisfies (a) and (b). For each $v \in X$ with $P(v) \neq \top$ and each clause $j$, condition (a) provides a witness $u \in A_{v,j} \cap X$; fix one. Order $X$ by any topological order of $D(C)$ restricted to $X$ (possible by well-foundedness) and place the courses one per term in that order. Each witness of $v$ precedes $v$, so (F2) holds; free capacity gives (F1); the taken set is $X$, so (b) gives completeness. Hence $X \in \mathrm{Comp}(C)$.
\end{proof}

\begin{proposition}[Na\"ive transfer is vacuous]\label{prop:vacuous}
Let $C$ be well-founded with free capacity and $r_g \le \credits(S_g)$ for every group. Then $V \in \mathrm{Comp}(C)$; in particular $L(C) \neq \emptyset$, so under free capacity the two objectives of Definition~\ref{def:obj} are always finite and their $+\infty$ cases arise only under a binding capacity. Consequently, for any curricula $C_1, C_2$ and any matching, some completion of $C_2$ contains the image of every completion of $C_1$, so the requirement that the receiver merely \emph{honour} transferred credit while completing its degree is satisfied trivially.
\end{proposition}

\begin{proof}
Taking $X = V$ in Lemma~\ref{lem:comp}: every clause $A_{v,j}$ is a nonempty subset of $V$, so $V \models P(v)$ for all $v$, giving (a); and $\credits(V \cap S_g) = \credits(S_g) \ge r_g$ gives (b). Thus $V \in \mathrm{Comp}(C_2)$, and the image of $\taken(\pi_1)$ under any matching is contained in $V$ for every plan $\pi_1$.
\end{proof}

Proposition~\ref{prop:vacuous} shows that the substantive articulation question is not whether the receiver can complete while retaining transferred courses, which it always can by taking everything, but whether completing the sending program already \emph{constitutes} completion of the receiving one. This is the direction that matters for recognition of one degree by the standard of another.

\subsection{Program subsumption}

Throughout this subsection both curricula have free capacity, so that Lemma~\ref{lem:comp} characterizes their completions; the bounded-capacity case is discussed at the end.

\begin{definition}[Subsumption]\label{def:subsume}
Let $C_1$ and $C_2$ be free-capacity curricula with a matching $\mu \subseteq V_1 \times V_2$, and for $X \subseteq V_1$ write $\mu(X) = \{ w \in V_2 : (v,w) \in \mu,\ v \in X\}$. We say $C_1$ \emph{subsumes} $C_2$ under $\mu$, written $C_1 \sqsubseteq_\mu C_2$, if $\mu(X) \in \mathrm{Comp}(C_2)$ for every $X \in \mathrm{Comp}(C_1)$; that is, every way of completing $C_1$ already completes $C_2$ once credit is transferred.
\end{definition}

\begin{proposition}[Membership]\label{prop:submember}
For free-capacity curricula, deciding $C_1 \sqsubseteq_\mu C_2$ is in \textup{coNP}.
\end{proposition}

\begin{proof}
A certificate of non-subsumption is a set $X \subseteq V_1$ together with the claims that $X \in \mathrm{Comp}(C_1)$ and $\mu(X) \notin \mathrm{Comp}(C_2)$. Under free capacity, Lemma~\ref{lem:comp} makes both claims checkable in polynomial time, testing prerequisite-closure and the requirement thresholds in each curriculum. Hence non-subsumption is in \textup{NP} and subsumption in \textup{coNP}. (The membership argument relies on the free-capacity characterization of $\mathrm{Comp}$; under a bounded capacity, deciding membership in $\mathrm{Comp}(C)$ is itself a capacity-scheduling question, and the complexity of subsumption is left open.)
\end{proof}

\begin{theorem}[Subsumption is coNP-complete]\label{thm:subhard}
Deciding $C_1 \sqsubseteq_\mu C_2$ is \textup{coNP}-complete, already when $C_1$ and $C_2$ are prerequisite-free (hence conjunctive) with free capacity and $\mu$ is a bijection.
\end{theorem}

\begin{proof}
Membership is Proposition~\ref{prop:submember}. For hardness we reduce from the complement of \textsc{Subset Sum}, which is coNP-complete because \textsc{Subset Sum} is NP-complete \citep{gj1979}. An instance of \textsc{Subset Sum} is a list of positive integers $a_1,\dots,a_N$ and a target $b$; we may assume $0 < b < \sum_i a_i$. Let $Z = \{z_1,\dots,z_N\}$ be a set of courses with $\credits(z_i) = a_i$ and no prerequisites, and take $\mu$ the identity on $Z$. Define
\[
  C_1 : R_1 = \{ (Z,\, b) \}, \qquad C_2 : R_2 = \{ (Z,\, b+1) \}, \qquad \kappa = \infty \text{ in both.}
\]
By Lemma~\ref{lem:comp}, since all courses are prerequisite-free, $\mathrm{Comp}(C_1) = \{ X \subseteq Z : \credits(X) \ge b\}$ and $\mathrm{Comp}(C_2) = \{ X \subseteq Z : \credits(X) \ge b+1 \}$. Therefore
\begin{align*}
  C_1 \sqsubseteq_\mu C_2
  &\iff \big(\forall X \subseteq Z:\ \credits(X) \ge b \Rightarrow \credits(X) \ge b+1\big) \\
  &\iff \neg\,\exists X \subseteq Z:\ \credits(X) = b,
\end{align*}
the last equivalence because the weights are integers, so $\credits(X) \ge b$ and $\credits(X) < b+1$ together mean $\credits(X) = b$. Hence $C_1 \sqsubseteq_\mu C_2$ holds if and only if the \textsc{Subset Sum} instance is a no-instance. The construction is polynomial, so the complement of \textsc{Subset Sum} reduces to subsumption, and subsumption is \textup{coNP}-hard.
\end{proof}

\begin{proposition}[Strong hardness of subsumption]\label{prop:substrong}
Subsumption is \textup{coNP}-hard already for prerequisite-free, free-capacity curricula in which every credit is one, so its hardness does not rest on large integers.
\end{proposition}

\begin{proof}
Reduce from the complement of \textsc{Set Cover}. Given a universe $U$, a family $F_1,\dots,F_s \subseteq U$, and an integer $t$, let both curricula share the prerequisite-free unit-credit courses $\{x_1,\dots,x_s\}$, with $\mu$ the identity. Give $C_1$ the element groups $(S_e, 1)$, $S_e = \{ x_i : e \in F_i\}$, so by Lemma~\ref{lem:comp} $\mathrm{Comp}(C_1)$ is exactly the covers of $U$, and give $C_2$ those same groups together with one further group $(\{x_1,\dots,x_s\},\, t+1)$, so $\mathrm{Comp}(C_2)$ is the covers of size at least $t+1$. Then $C_1 \sqsubseteq_\mu C_2$ holds iff every cover of $U$ has size at least $t+1$, that is, iff $U$ has no cover of size at most $t$. Since \textsc{Set Cover} is NP-hard and carries only unit numbers, its complement is \textup{coNP}-hard with unit credits, and subsumption inherits this.
\end{proof}

\subsection{Consensus prerequisite recovery}

We treat the conjunctive case, in which the prerequisite of each course is a set of required predecessors, so a curriculum's prerequisite structure is its support digraph. For an assignment $P$ let $A(P) \subseteq V \times V$ be the corresponding arc set.

\begin{definition}[Consensus]\label{def:consensus}
Given conjunctive prerequisite assignments $P_1,\dots,P_K$ over a common course set $V$, which need not themselves be acyclic since they are noisy observations (the corpus of Section~\ref{sec:empirical} contains such cycles as parsing artifacts), the \emph{consensus} problem asks for an \emph{acyclic} assignment $P$ minimizing
\[
  \Delta(P) \;\eqdef\; \sum_{k=1}^{K} \big| A(P) \,\triangle\, A(P_k) \big|,
\]
where $\triangle$ is symmetric difference; the decision version asks whether $\Delta(P) \le D$ for a given bound $D$.
\end{definition}

\begin{theorem}[Consensus is NP-complete]\label{thm:consensus}
The consensus decision problem is NP-complete, already for a single input assignment ($K = 1$).
\end{theorem}

\begin{proof}
\emph{Membership.} A candidate $P$ is an assignment over $V$, of polynomial size; acyclicity of $A(P)$ and the value $\Delta(P)$ are checked in polynomial time, so the problem is in NP.

\emph{Hardness.} We reduce from \textsc{Minimum Feedback Arc Set}, which is NP-complete \citep{karp1972}: given a digraph $G = (V,E)$ and an integer $f$, decide whether some set of at most $f$ arcs meets every directed cycle, equivalently whether $G$ has an acyclic subgraph with at least $|E| - f$ arcs. Take $K = 1$ and let $P_1$ be the conjunctive assignment with $A(P_1) = E$. For any acyclic assignment $P$ with arc set $A \eqdef A(P)$,
\[
  \Delta(P) = |A \triangle E| = |E \setminus A| + |A \setminus E|.
\]
Replacing $A$ by $A \cap E$ removes the second term without creating a cycle and without increasing the first, so an optimal $P$ satisfies $A \subseteq E$, whence $\Delta(P) = |E \setminus A| = |E| - |A|$. Minimizing $\Delta$ over acyclic $A \subseteq E$ therefore maximizes $|A|$ over acyclic subgraphs of $G$, and
\[
  \min_{P \text{ acyclic}} \Delta(P) \;=\; |E| - \max\{\, |A| : A \subseteq E,\ (V,A) \text{ acyclic}\,\}
\]
equals the minimum feedback arc set size of $G$. Thus there is an acyclic $P$ with $\Delta(P) \le f$ if and only if $G$ has a feedback arc set of size at most $f$. The reduction is polynomial.
\end{proof}

\begin{remark}
The two results delimit the model's comparison questions exactly: subsumption is complete for \textup{coNP} and consensus for \textup{NP}, and both remain hard under the strongest structural restrictions, prerequisite-free curricula for subsumption and a single input for consensus. The source of consensus hardness is isolated cleanly, namely the acyclicity requirement on the output: without it the arc-wise majority vote is optimal and the problem is trivial, so the difficulty is entirely that of reconciling conflicting prerequisites into a well-founded whole. The single-input reduction relies on a cyclic observation and is cleanest as a hardness embedding; the empirically motivated case, in which $K \ge 3$ individually acyclic curricula are reconciled and cycles appear only in their aggregate, is the setting the application meets. Theorem~\ref{thm:consensus} establishes NP-completeness of the problem in general, and we conjecture that hardness persists for this restricted multi-input case, an acyclic-aggregation question we leave open. The approximability of consensus remains a natural next question, inheriting the standing bounds for feedback arc set; the parallel question for subsumption, whether its hardness rests on large integers, is settled in the negative by Proposition~\ref{prop:substrong}, which gives \textup{coNP}-hardness already with unit credits.
\end{remark}

\section{Conclusion}\label{sec:summary}

The model of Definitions~\ref{def:curriculum}--\ref{def:obj} reduces a curriculum to three structural commitments beyond a bare precedence graph, namely the Boolean form of prerequisites, the credit-threshold form of requirements, and the per-term capacity. The theorems then attribute the difficulty of degree planning to those commitments precisely. Recognition is linear (Proposition~\ref{prop:membership}). Time to degree is polynomial whenever capacity is unbounded, in full generality over disjunction and electives (Theorem~\ref{thm:timefree}), and becomes strongly NP-hard the moment capacity binds, already without prerequisites (Theorem~\ref{thm:capacityhard}). Load is complementary: polynomial on the conjunctive and mandatory fragment (Proposition~\ref{prop:loadconj}) and strongly NP-hard through either of two independent commitments, disjunction (Theorem~\ref{thm:loadhard}) or overlapping electives (Theorem~\ref{thm:loadover}). The two objectives thus have disjoint sources of hardness (Theorem~\ref{thm:duality}): capacity governs time, and disjunction or electives govern load. Finally, the delay-factor component of the curricular-complexity metric is revealed as a polynomially computable upper bound on time to degree, exact on the conjunctive fragment and loose by the disjunctive slack elsewhere (Proposition~\ref{prop:metric}), which is the formal counterpart of the empirical correction measured across the corpus. Beyond the two objectives, the model's comparison questions are complete for their natural classes: under free capacity, program subsumption is coNP-complete (Theorem~\ref{thm:subhard}), through a subset-sum obstruction that survives even prerequisite-free curricula, and consensus prerequisite recovery is NP-complete (Theorem~\ref{thm:consensus}), its hardness residing wholly in the acyclicity demanded of the reconciled relation. Instantiated on a corpus of twenty-two universities (Section~\ref{sec:empirical}), the model places real curricula on the tractable side of every boundary but capacity: nearly nine in ten courses are purely conjunctive; the curriculum-level disjunctive slack is zero throughout, not because chains are shallow but because the deepest chain of every program runs through mandatory conjunctive courses that no alternative shortens; and the course-level slack is positive on roughly one course in eleven, concentrated at the institutions publishing the most alternatives. Prerequisite chains are in any case shallow enough that capacity, not prerequisite logic, is the constraint governing timely completion.

\section*{Acknowledgment}
The authors gratefully acknowledge the College of Computing and Artificial Intelligence and the Research and Innovation Office at the United Arab Emirates University for their support, which covered the article processing charge for open-access publication and provided the institutional resources, including computational facilities, library access, and infrastructural support, that underpinned the analysis reported in this work.

\setlength{\bibsep}{2pt plus 1pt}% compact reference list (avoids a one-entry final page)
\bibliographystyle{plainnat-doi}
\bibliography{ref}

\end{document}